%% file: main.tex
\documentclass[11pt,letterpaper]{article}  

\usepackage[dvipsnames]{xcolor}
\usepackage{comment}
\usepackage{tikz}
\usetikzlibrary{external}
\usetikzlibrary{mindmap,trees}
\usetikzlibrary{arrows}
\usetikzlibrary{calc}
\usepackage{adjustbox}
\usepackage{subcaption}
\usepackage{epsfig}
\usepackage{epstopdf}

\usepackage{amsthm}

\newtheorem{corollary}{Corollary}
\newtheorem{lemma}{Lemma}
\newtheorem{remark}{Remark}
\newtheorem{definition}{Definition}
\newtheorem{proposition}{Proposition}
\newtheorem{example}{Example}
\newtheorem{assumption}{Assumption}
\usepackage{amsmath,amsfonts,amssymb,color}
\DeclareMathOperator*{\argmin}{argmin}
\usepackage{psfrag}
\usepackage{algorithm}
\usepackage{algpseudocode}
\usepackage{array}
\usepackage{physics}
\usepackage{cite}

\newcommand{\parinaz}[1]{\textcolor{purple}{PN: #1}}

\newcommand{\Rb}{\mathbb{R}}

\newcommand{\EE}{\mathcal{E}}

\newcommand{\DD}{\mathcal{D}}

\usepackage[margin=1in]{geometry}                                                    
\usepackage[hidelinks]{hyperref}
\hypersetup{
    colorlinks,
    linkcolor={blue!},
    citecolor={blue!}
}

\title{\LARGE \bf Behavioral and Game-Theoretic Security Investments in Interdependent Systems Modeled by Attack Graphs} 
\author{Mustafa Abdallah,~Parinaz Naghizadeh,~Ashish R. Hota,~Timothy Cason,\\~Saurabh Bagchi, and~Shreyas Sundaram
\thanks{This research was supported by grant CNS-1718637 from the National Science Foundation. Mustafa Abdallah, Saurabh Bagchi, and Shreyas Sundaram are with the School of Electrical and Computer Engineering at Purdue University, West Lafayette, Indiana, USA, 47907. Email: {\tt \{abdalla0,sbagchi,sundara2\}@purdue.edu}. Parinaz Naghizadeh is with the Integrated Systems Engineering Department and the Electrical and Computer Engineering Department, Ohio State University, USA. Email: {\tt{naghizadeh.1@osu.edu}}. Ashish R. Hota is with the Department of Electrical Engineering, Indian Institute of Technology (IIT), Kharagpur, India. Email: {\tt{ahota@ee.iitkgp.ac.in}}. Timothy Cason is with the Krannert School of Management at Purdue University. Email: {\tt{cason@purdue.edu}}. Corresponding author: Shreyas Sundaram. A preliminary version of this paper appears in the Proceedings of the American Control Conference 2019.}.%
}

\begin{document}
\maketitle

\begin{abstract}
We consider a system consisting of multiple interdependent assets, and a set of defenders, each responsible for securing a subset of the assets against an attacker. The interdependencies between the assets are captured by an attack graph, where an edge from one asset to another indicates that if the former asset is compromised, an attack can be launched on the latter asset.  Each edge has an associated probability of successful attack, which can be reduced via security investments by the defenders.  In such scenarios, we investigate the security investments that arise under certain features of human decision-making that have been identified in behavioral economics.  In particular, humans have been shown to perceive probabilities in a nonlinear manner, typically overweighting low probabilities and underweighting high probabilities. We show that suboptimal investments can arise under such weighting in certain network topologies. We also show that pure strategy Nash equilibria exist in settings with multiple (behavioral) defenders, and study the inefficiency of the equilibrium investments by behavioral defenders compared to a centralized socially optimal solution.
\end{abstract}

\input{Introduction}
\input{Background}
\input{Behavioral_Game_Classes}
\input{Optimal_investments_single_player}

\input{OptimalInvestmentsPNE}

\input{example}

\input{Conclusion}

\bibliographystyle{unsrt}
\bibliography{References}

\end{document}

%% file: Introduction.tex
\section{Introduction}

Modern cyber-physical systems (CPS) are increasingly facing attacks by sophisticated adversaries. These attackers are able to identify the susceptibility of different targets in the system and strategically allocate their efforts to compromise the security of the network. In response to such intelligent adversaries, the operators (or defenders) of these systems also need to allocate their often limited security budget across many assets to best mitigate their vulnerabilities. This has led to significant research in understanding how to better secure these systems, with game-theoretical models receiving increasing attention  due to their ability to systematically capture the interactions of strategic attackers and defenders \cite{humayed2017cyber,laszka2015survey,alpcan2010network,sanjab2016data, miao2018hybrid, milosevic2019network, brown2018security, riehl2017centrality}. 


In the context of large-scale interdependent systems, adversaries often use stepping-stone attacks to exploit vulnerabilities within the network in order to compromise a particular target \cite{zonouz2012scpse}.  Such threats can be captured via the notion of {\it attack graphs} that  represent all possible paths that attackers may have to reach their targets within the CPS  \cite{homer2013aggregating}. The defenders in such systems are each responsible for defending some subset of the assets  \cite{laszka2015survey,naghizadeh2016opting} with their limited resources. These settings have been explored under various assumptions on the defenders and attackers \cite{la2016interdependent,naghizadeh2016opting,hota2018game}. 

In much of the existing literature, the defenders and attackers are modeled as fully rational decision-makers who choose their actions to maximize their expected utilities. However, a large body of work in behavioral economics has shown that humans consistently deviate from such classical models of decision-making \cite{kahneman1979prospect, dhami2016foundations,barberis2013thirty}.  A seminal model capturing such deviations is {\it prospect theory} (introduced by  Kahneman and Tversky in \cite{kahneman1979prospect}), which shows that humans perceive gains, losses, and probabilities in a skewed (nonlinear) manner, typically overweighting low probabilities and underweighting high probabilities.  Recent papers have studied the implications of prospect theoretic
preferences in the context of CPS security and robustness \cite{7544460,hota2016fragility,9030279}, energy consumption decisions in the smart grid \cite{etesami2018stochastic}, pricing in communication networks \cite{yang2015prospect}, and network interdiction games \cite{sanjab2017prospect}.


In this paper, we consider the scenario where each (human) defender misperceives the probabilities of successful attack in the attack graph.\footnote{
While existing literature on behavioral aspects of information security, such as \cite{baddeley2011information,anderson2012security,schneier2008psychology} rely on human subject experiments and more abstract decision-making models, we consider the more concrete framework of attack graphs in our analysis. This framework allows for a mapping from existing vulnerabilities to potential attack scenarios. Specifically, one model that is captured by our formulation is to define vulnerabilities by CVE-IDs \cite{martin2001managing}, and assign attack probabilities using the Common Vulnerability Scoring System (CVSS) \cite{mell2006common}.} We characterize the impacts of such misperceptions on the security investments made by each defender. In contrast with prior work on prospect theoretic preferences in the context of CPS security, \cite{7544460} 
which assumed that each defender is only responsible for the security of a single node, we consider a more general case where each defender is responsible for a subnetwork (i.e., set of assets). Furthermore, each defender can also invest in protecting the assets of other defenders, which may be beneficial in interdependent CPS where the attacker exploits paths through the network to reach certain target nodes.


Specifically, we build upon the recent work \cite{hota2018game} where the authors studied a game-theoretic formulation involving attack graph models of interdependent systems and multiple defenders. The authors showed how to compute the optimal defense strategies for each defender using a convex optimization problem. However, they did not investigate the characteristics of optimal investments and the impacts of behavioral biases of the defenders which are the focus of the present work.

We introduce the attack-graph based security game framework in Section \ref{sec: framework}, followed by the behavioral security game setting in Section \ref{sec:behavioralclasses}. Under appropriate assumptions on the probabilities of successful attack on each edge, we establish the convexity of the perceived expected cost of each defender and prove the existence of a pure Nash equilibrium (PNE) in this class of games. 

We primarily investigate the security investments when users with such behavioral biases act in isolation (Section \ref{sec:single_optimal_inves_dec}) as well as in a game-theoretic setting (Section \ref{sec: PNE}). As a result, we find certain characteristics of the security investments under behavioral decision making that could not have been predicted under classical notions of decision-making (i.e., expected cost minimization) considered in prior work \cite{hota2018game}. In particular, we show that nonlinear probability weighting can cause defenders to invest in a manner that increases the vulnerability of their assets to attack. Furthermore, we illustrate the impacts of having a mix of defenders (with heterogeneous levels of probability weighting bias) in the system, and show that the presence of defenders with skewed perceptions of probability can in fact \emph{benefit} the non-behavioral defenders in the system.



We then propose a new metric, \emph{Price of Behavioral Anarchy  (PoBA)}, to capture the inefficiency of the equilibrium investments made by behavioral decision-makers compared to a centralized (non-behavioral)  socially optimal solution, and provide tight bounds for the PoBA. We illustrate the applicability of the proposed framework in a case study involving a distributed energy resource failure scenario, DER.1, identified by the US National Electric Sector Cybersecurity Organization Resource (NESCOR) \cite{jauhar2015model} in Section~\ref{sec:example}.

This paper extends the conference version of this work \cite{8814307} in the following manner:

$\bullet$ We rigorously prove the uniqueness of optimal investment decisions for behavioral defenders, and show that Behavioral Security Games can have multiple PNEs in general.

$\bullet$ We quantify the inefficiency of the Nash equilibria by defining the notion of PoBA, and provide (tight) bounds on it. 

$\bullet$ We illustrate the theoretical findings via a case study.

%% file: Background.tex
\section{The Security Game Framework}\label{sec: framework}

In this section, we describe our general security game framework, including the attack graph and the characteristics of the attacker and the defenders. An overview of our model is shown in Figure~\ref{fig:overview_model}.

\subsection{Attack Graph}
We represent the assets in a CPS as nodes of a directed graph $ G=(V,\mathcal{E}) $ where each node $v_{i} \in V $ represents an asset. A directed edge $(v_{i},v_{j}) \in \mathcal{E}$ means that if $v_{i}$ is successfully attacked, it can be used to launch an attack on $v_{j}$. 

The graph contains a designated source node $ v_{s} $ (as shown in Figure~\ref{fig:overview_model}), which is used by an attacker to begin her attack on the network. Note that $v_s$ is not a part of the network under defense; rather it is an entry point that is used by an attacker to begin her attack on the network.\footnote{If there are multiple nodes where the attacker can begin her attack, then we can add a virtual node $v_s$, and add edges from this virtual node to these other nodes with attack success probability $1$ without affecting our formulation.}

For a general asset $ v_{t} \in V$, we define $\mathcal{P}_{t}$ to be the set of directed paths from the source $ v_{s} $ to $v_{t}$ on the graph, where a path $ P \in \mathcal{P}_{t}$ is a collection of edges $ \lbrace{(v_{s}, v_{1}), (v_{1}, v_{2}), . . . , (v_{k}, v_{t})\rbrace}$. For instance, in Figure~\ref{fig:overview_model}, there are two attack paths from $v_s$ to $v_t$.


Each edge $(v_i, v_j) \in \mathcal{E}$ has an associated weight $p_{i,j}^0 \in (0,1]$, which denotes the probability of successful attack on asset $v_j$ starting from $v_i$ in the absence of any security investments.\footnote{
In practice, CVSS \cite{mell2006common} can be used for estimating initial probabilities of attack (for each edge in our setting). For example, \cite{homer2013aggregating} takes the Access Complexity (AC) sub-metric in CVSS (which takes values in \{low, medium, high\}, representing the complexity of exploiting the vulnerability) and maps it to a probability of exploit (attack) success. The more complex it is to exploit a vulnerability, the less likely an attacker will succeed. Similarly,  \cite{8005480} provides methods and tables to estimate the probability of successful attack from CVSS metrics.} 

 We now describe the defender and adversary models in the following two subsections.




\begin{figure}
\begin{center}
  \includegraphics[width=0.5\columnwidth]{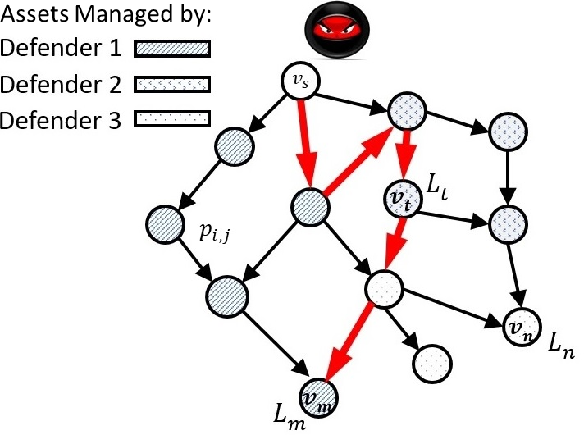}
  \caption{Overview of the interdependent security game framework. This CPS consists of three interdependent defenders. An attacker tries to compromise critical assets starting from $v_s$. 
  }
  \label{fig:overview_model}
\end{center}
\end{figure}

\subsection{Strategic Defenders}
Let $ \mathcal{D} $ be the set of all defenders of the network. Each defender $ D_{k} \in \mathcal{D} $ is responsible for defending a set $ V_{k} \subseteq V \setminus \lbrace{v_s\rbrace} $ of assets. For each compromised asset $ v_{m} \in V_{k} $, defender $ D_{k} $ will incur a financial {\it loss} $L_{m} \in [0,\infty) $. For instance, in the example shown in Figure \ref{fig:overview_model}, there are three defenders with assets shown in different shades, and the loss values of specific nodes are indicated.


To reduce the attack success probabilities on edges interconnecting assets inside the network, a defender can allocate security resources on these edges.\footnote{Note that $v_s$ does not have any incoming edges, and hence, it can not be defended.} We assume that each defender $D_k$ has a security budget $B_k \in [0,\infty)$. Let $x^k_{i,j}$ denote the security investment of defender $D_k$ on the edge $(v_i,v_j)$. We define
\begin{equation} \label{eq:defense_strategy_space}
X_k := \{x_k \in \Rb^{|\EE|}_{\geq 0} | \mathbf{1}^T x_k \leq B_k\};
\end{equation}
thus $X_k$ is the set of feasible investments for defender $D_k$ and it consists of all possible non-negative investments on the edges of the graph such that the sum of these investments is upper bounded by $B_k$. We denote any particular vector of investments by defender $D_k$ as $x_k \in X_k$. Each entry of $x_k$ denotes the investment on an edge.

Let $ \mathbf{x} = \left[ x_{1}, x_{2}, \ldots, x_{\mathcal{|D|}} \right] $ be a joint defense strategy of all defenders, with $ x_{k} \in X_{k} $ for defender $ D_{k} $; thus, $\mathbf{x} \in \mathbb{R^{|\mathcal{D}| |\mathcal{E}|}_{\geq \texttt{0}}}$.
Under a joint defense strategy $\mathbf{x}$, the total investment on edge $(v_i,v_j)$ is  
$x_{i,j}\triangleq \sum_{D_k\in \mathcal{D}} x^k_{i,j} $.
%
Let $p_{i,j}:\mathbb{R}_{\geq 0}\rightarrow [0,1]$ be a function mapping the total investment $x_{i,j}$ to an attack success probability, with $p_{i,j}(0) = p_{i,j}^0$. In particular, $p_{i,j}(x_{i,j})$ is the conditional probability that an attack launched from $v_i$ to $v_j$ succeeds, given that $v_i$ has been successfully compromised.
%
%

\subsection{Adversary Model and Defender Cost Function}

In networked cyber-physical systems (CPS), there are a variety of adversaries with different capabilities that are simultaneously trying to compromise different assets. We consider an attacker model that uses stepping-stone attacks \cite{zonouz2012scpse}. In particular, for each asset in the network, we consider an attacker that starts at the entry node $v_s$ and attempts to compromise a sequence of nodes (moving along the edges of the network) until it reaches its target asset. If the attack at any intermediate node is not successful, the attacker is detected and removed from the network. Note that our formulation allows each asset to be targeted by a different attacker, potentially starting from different points in the network.

In other words, after the defense investments have been made, then for each asset in the network, the attacker chooses the path with the highest probability of successful attack for that asset (such a path is shown in red in Figure \ref{fig:overview_model}). Such attack models (where the attacker chooses one path to her target asset) have previously been considered in the literature (e.g., \cite{jain2011double,brown2011defender}).

To capture this, for a given set of security investments by the defenders, we define the {\it vulnerability} of a node $v_m \in V$ as $\displaystyle \max_{P \in \mathcal{P}_m} \prod_{(v_i,v_j) \in P} p_{i,j}(x_{i,j})$, where $\mathcal{P}_m$ is the set of all directed paths from the source $v_s$ to asset $v_m$; note that for any given path $P \in \mathcal{P}_m$, the probability of the attacker successfully compromising $v_m$ by taking the path $P$ is $\displaystyle \prod_{(v_i,v_j)\in P} p_{i,j}(x_{i,j})$, where $p_{i,j}(x_{i,j})$ is the conditional probability defined at the end of Section II-B. 
In other words, the vulnerability of each asset is defined as the maximum of the attack probabilities among all available paths to that asset.

The goal of each defender $ D_{k} $ is to choose her investment $x_k \in X_k$ in order to minimize the expected cost defined as
\begin{equation}\label{eq:defender_utility}
\hat{C}_{k}(x_{k},\mathbf{x}_{-k}) = \sum_{v_{m} \in V_{k}} L_{m} \hspace{0.3mm} \Big( \hspace{0.3mm} \underset{P \in \mathcal{P}_{m}}{\text{max}}\prod_{(v_{i},v_{j}) \in P} p_{i,j}({x_{i,j}}) \hspace{0.3mm} \Big)
\end{equation}
subject to $x_{k} \in X_{k}$, and where $ \mathbf{x}_{-k} $ is the vector of investments by defenders other than $ D_{k}$. Thus, each defender chooses her investments in order to minimize the vulnerability of her assets, i.e., the highest probability of attack among all available paths to each of her assets.\footnote{This also models settings where the specific path taken by the attacker or the attack plan is not known to the defender apriori, and the defender seeks to make the most vulnerable path to each of her assets as secure as possible.}

In the next section, we review certain classes of probability weighting functions that capture human misperception of probabilities. Subsequently, we introduce such functions into the above security game formulation, and study their impact on the investment decisions and equilibria.

\section{Nonlinear Probability Weighting and the Behavioral Security Game}\label{sec:behavioralclasses}

\subsection{Nonlinear Probability Weighting}
The behavioral economics and psychology literature has shown that humans consistently misperceive probabilities by overweighting low probabilities and underweighting high probabilities \cite{kahneman1979prospect,prelec1998probability}.  More specifically, humans perceive a ``true'' probability $p \in [0,1]$ as $w(p) \in [0,1]$, where $w(\cdot)$ is a probability weighting function.  A commonly studied probability weighting function was proposed by Prelec in \cite{prelec1998probability}, and is given by
\begin{equation}\label{eq:prelec}
w(p) = \exp\Big[-(-\log(p)\hspace{0.2mm})^{\alpha}\hspace{0.5mm} \Big],  \hspace{3mm} p\in [0,1],
\end{equation}
where $\alpha \in (0,1]$ is a parameter that controls the extent of overweighting and underweighting.  When $\alpha = 1$, we have $w(p) = p$ for all $p \in [0,1]$, which corresponds to the situation where probabilities are perceived correctly.  Smaller values of $\alpha$ lead to a greater amount of overweighting and underweighting, as illustrated in 
Figure \ref{fig:Prelec Probability weighting function}.  
Next, we incorporate this probability weighting function into the security game defined in the last section, and define the Behavioral Security Game that is the focus of this paper.
\begin{figure}
\begin{center}
  \includegraphics[width=0.5\columnwidth]{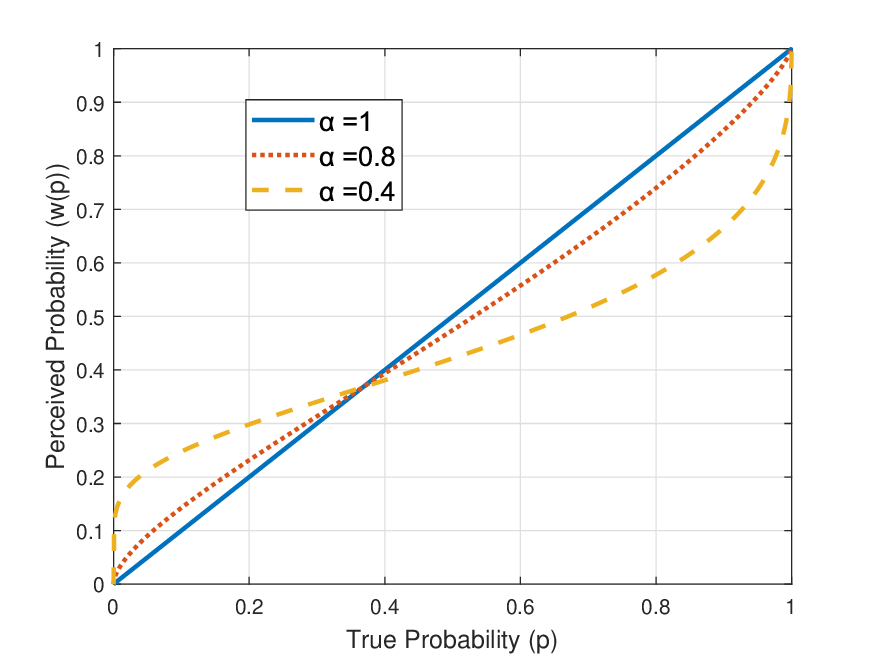}
  \caption{{Prelec probability weighting function \eqref{eq:prelec} which transforms true probabilities $ p $  into perceived probabilities $ w(p) $. The parameter $\alpha$ controls the extent of overweighting and underweighting.} 
  }
  \label{fig:Prelec Probability weighting function}
\end{center}
\end{figure}

%% file: Behavioral_Game_Classes.tex
\subsection{The Behavioral Security Game} 
Recall that each defender seeks to protect a set of assets, and the probability of each asset being successfully attacked is determined by the corresponding probabilities on the edges that constitute the paths from the source node to that asset.  This motivates a broad class of games that incorporate probability weighting, as defined below.

\begin{definition}\label{def:Edge-based Behavioral_Security_Game}
We define a {\it Behavioral Security Game} as a game between different defenders in an interdependent network, where each defender misperceives the attack probability on each edge according to the probability weighting function defined in \eqref{eq:prelec}. Specifically, the perceived attack probability by a defender $D_{k}$ on an edge $ (v_{i}, v_{j})$  is given by 
\begin{equation}
w_{k}( p_{i,j}(x_{i,j})) = \exp\Big[-(-\log(p_{i,j}(x_{i,j}))\hspace{0.2mm})^{\alpha_{k}}\hspace{0.5mm} \Big],
\label{eq:prelec_k}
\end{equation}
where $p_{i,j}(x_{i,j})\in [0,1]$ and $\alpha_k \in (0,1].$
\end{definition}

\begin{remark}
\normalfont The subscript $k$ in $\alpha_{k}$ and $w_{k}(\cdot)$ allows each defender in the Behavioral Security Game to have a different level of misperception.  We will drop the subscript $k$ when it is clear from the context.$\hfill \blacksquare$

\end{remark}

Incorporating this into the cost function \eqref{eq:defender_utility}, each defender $ D_{k} $ seeks to minimize her {\it perceived expected cost} 
\begin{equation}\label{eq:defender_utility_edge}
C_{k}(x_k, \mathbf{x}_{-k}) \!\!= \!\!\sum_{v_{m} \in V_{k}} \!\!L_{m} \left(\!\!\underset{P \in \mathcal{P}_{m}}{\text{max}}\!\prod_{(v_{i},v_{j}) \in P} w_{k}\left( p_{i,j}(x_{i,j}) \right) \right)\!.
\end{equation}
Thus, our formulation complements the existing decision-making models based on vulnerability and cost by incorporating certain behavioral biases in the cost function.

\begin{remark}
\normalfont In addition to misperceptions of probabilities, empirical evidence shows that humans perceive costs differently from their true values. In particular, humans (i) compare uncertain outcomes with a reference utility or cost, (ii) exhibit risk aversion in gains and risk seeking behavior in losses, and (iii) overweight losses compared to gains (loss aversion). A richer behavioral model, referred to as cumulative prospect theory [10], incorporates all these aspects in its cost function. However, in the setting of this paper, this richer model does not significantly change the cost functions of the defenders. Specifically, the attack on an asset is either successful or it is not. If the reference cost is zero for each asset (i.e., the default state where the asset is not attacked successfully), then successful attack constitutes a loss, and the index of loss aversion only scales the constant $L_m$ by a scalar without changing the dependence of the cost function on the investments. $\hfill \blacksquare$
\end{remark}

\subsection{Assumptions on the Probabilities of Successful Attack}
The shape of the probability weighting function \eqref{eq:prelec} presents several challenges for analysis.  In order to maintain analytical tractability, we make the following assumption on the probabilities of successful attack on each edge.

\begin{assumption}\label{ass:attack_prob} 
For every edge $(v_i,v_j)$, the probability of successful attack $p_{i,j}(x_{i,j})$ is log-convex\footnote{This is a common assumption in the literature. In particular, \cite{baryshnikov2012security} shows that log-convexity of the attack probability functions is a necessary and sufficient condition for the optimal security investment result of the seminal paper \cite{gordon2002economics} to
hold.
}, strictly decreasing, and twice continuously differentiable for $x_{i,j} \in [0,\infty)$. 
\end{assumption}

One particular function satisfying the above conditions is
\begin{equation}\label{eq:expon_prob_func}
p_{i,j}(x_{i,j})= p_{i,j}^0\exp(-x_{i,j}). 
\end{equation}
Such probability functions fall within the class commonly considered in security economics (e.g., \cite{gordon2002economics}), and we will specialize our analysis to this class for certain results in the paper.  For such functions, the (true) attack success probability of any given path $P$ from the source to a target $v_t$ is given by
\begin{equation}
\prod_{(v_m,v_n) \in P}p_{m,n}(x_{m,n}) 
= \Big(\prod_{(v_m,v_n) \in P}p_{m,n}^0\Big)\exp\Big(-  \sum_{(v_m,v_n)\in P}x_{m,n} \Big).
\label{eq:path_probability}
\end{equation}
Thus, the probability of successful attack on a given path decreases exponentially with the sum of the investments on all edges on that path by all defenders.  

\begin{remark}
\normalfont The paper \cite{hota2018game} studied this same class of security games for the case of non-behavioral defenders (i.e., with $\alpha_k = 1, \forall D_k \in \mathcal{D}$).  For that case, with  probability functions given  by \eqref{eq:expon_prob_func}, \cite{hota2018game} showed that the optimal investments for each defender can be found by solving a convex optimization problem.  Suitable modifications of the same approach to account for the parameter $\alpha_k$ will also work for determining the optimal investments by the behavioral defenders in this paper. We omit the details in the interest of space.$\hfill \blacksquare$
\end{remark}

%% file: Optimal_investments_single_player.tex
\section{Properties of the Optimal Investment Decisions By a Single Defender}\label{sec:single_optimal_inves_dec}
We  start our analysis of the impact of behavioral decision-making by considering settings with only a single defender (i.e., $|\mathcal{D}| = 1$).  In particular, we will establish certain properties of the defender's cost function \eqref{eq:defender_utility_edge}, and subsequently identify properties of the defender's optimal investment decisions under behavioral (i.e., $\alpha < 1$) and non-behavioral (i.e., $\alpha = 1$) decision-making. This setting will help in understanding the actions (i.e., best responses) of each player in  multi-defender Behavioral Security Games, which we will consider in the next section. In this section, we will refer to the defender as $D_k$, and drop the vector $\mathbf{x}_{-k}$ from the arguments.  

\subsection{Convexity of the Cost Function}
We first prove the convexity of the defender's cost function. To do so, we start with the following result.



\begin{lemma} \label{lemma:perceived_convex}
For $\alpha_k \in (0,1)$ and  $(v_i,v_j) \in \mathcal{E}$, let $h(x_{i,j}) \triangleq (- \log(p_{i,j}(x_{i,j})))^{\alpha_k}$. Then, $h(x_{i,j})$ is strictly concave in $x_{i,j}$ for $x_{i,j} \in [0,\infty)$ under Assumption \ref{ass:attack_prob}. Moreover, $h(x_{i,j})$ is concave in $x_{i,j}$ for $\alpha_k \in (0,1]$.
\end{lemma}
\begin{proof}
For ease of notation, we drop the subscripts $i,j$, and $k$ in the following analysis. First, we focus on the case where $\alpha \in (0,1)$. Note from Assumption \ref{ass:attack_prob} that $0 < p (x) \leq 1$, and so $0 \leq -\log(p(x)) < \infty $ for all $x \in [0,\infty)$. 

Now, we prove that $h(x)$ is strictly concave:
\begin{align*}
h'(x) &= -\alpha (-\log(p(x)))^{\alpha-1} \frac{p'(x)}{p(x)} \\
h''(x) &= \alpha (\alpha-1) (-\log(p(x)))^{\alpha-2} \frac{(p'(x))^2}{(p(x))^2}\\
& + \alpha (-\log(p(x)))^{\alpha-1} \left[ \frac{(p'(x))^2 - p(x)p''(x)}{(p(x))^2} \right].
\end{align*}
From Assumption \ref{ass:attack_prob}, $p(x)$ is strictly decreasing and therefore $p'(x) < 0$. Thus, the first term on the R.H.S. of $h''(x)$ is strictly negative if $\alpha \in (0,1)$. Also, since $p(x)$ is twice-differentiable and log-convex with a convex feasible defense strategy domain $\mathbb{R}_{\ge 0}$, following \cite[Subsection 3.5.2]{boyd2004convex}, we have $(p'(x))^2 \leq p(x)p''(x)$, which ensures that the second term is non-positive. Therefore, $h(x)$ is strictly concave.

Finally, if $\alpha=1$, we have $h(x) = -\log(p(x))$, and since $p(x)$ is log-convex, $h(x)$ is concave.
\end{proof}

Using the above result, we now prove that the defender's cost function \eqref{eq:defender_utility_edge} is convex.

\begin{lemma}\label{lemma:edge_convex}
For all $\alpha_k\in (0,1]$ and under Assumption \ref{ass:attack_prob}, the cost function \eqref{eq:defender_utility_edge} of the defender $D_k$ is convex in the defense investment $x_{k}$.
\end{lemma} 
\begin{proof}
For each attack path $P$, define $ h_{P}(x_k) \triangleq \displaystyle \sum_{(v_i, v_j) \in P} (-\log(p_{i,j}(x_{i,j})))^{\alpha_k}$. Then, using the Prelec  function in \eqref{eq:prelec_k}, the cost in \eqref{eq:defender_utility_edge} is given by
\begin{equation*}
C_k(x_k) = \sum_{v_{m} \in V_{k}}  \hspace{0.3mm} L_m  \Big( \hspace{0.3mm} \underset{P \in \mathcal{P}_{m}}{\text{max}} \exp(-h_P(x_k)) \Big).
\end{equation*}
Note that $h_P(x_k)$ is separable and by Lemma \ref{lemma:perceived_convex}, each term in $h_P(x_k)$ is concave in a different variable (i.e., each term corresponds to a different edge $(v_i,v_j)$ in the attack path $P$).  Thus, $h_P(x_k)$ is concave in $x_k$, and so  $\exp(-h_P(x_k))$ is convex in $x_k$. 
Moreover, the maximum of  a set of convex functions is also convex \cite[Subsection 3.2.3]{boyd2004convex}. Finally, since $C_k(x_k)$ is a linear combination of convex functions, $C_k(x_k)$ is convex in $x_{k}$.
\end{proof}


\subsection{Uniqueness of Investments}
Having established the convexity of the defender's cost function \eqref{eq:defender_utility_edge}, we now observe the difference in the investment decisions made by behavioral and non-behavioral defenders.  In particular, we first show that the optimal investment decisions by a behavioral defender are unique, and then contrast that with the (generally) non-unique optimal investments for non-behavioral defenders.

\begin{proposition}\label{prop: uniqueness of investments}
Consider an attack graph $G = (V,\mathcal{E})$ and a defender $D_k$. Assume the probability of successful attack on each edge satisfies Assumption \ref{ass:attack_prob} 
and $\alpha_k \in (0,1)$ in the probability weighting function \eqref{eq:prelec_k}. 
Then, the optimal investments by defender $D_k$ to minimize \eqref{eq:defender_utility_edge} are unique.
\end{proposition}

\begin{proof}
Consider the defender's optimization problem for the cost function in \eqref{eq:defender_utility_edge}. 
Denote a path (after investments) to be a ``critical path'' of an asset if it has the highest probability of successful attack from the source to that asset (note that multiple paths can be critical). The ``value'' of a path is its probability of successful attack (product of perceived probabilities on each edge in the path).
 
We claim that in any optimal solution $x^{*}_{k}$, every edge that has a nonzero investment must belong to some critical path. Let $(v_a,v_b)$ be an edge that does not belong to any critical path\footnote{The proof holds even if there are multiple critical paths.} and suppose by contradiction that $x^{*}_{k}$ is an optimal solution of \eqref{eq:defender_utility_edge} in which the edge $(v_a,v_b)$ has a nonzero investment. Now, remove a sufficiently small nonzero investment $\epsilon$ from the edge $(v_a,v_b)$ and spread it equally among all of the edges of the critical paths. This reduces the total attack probability on the critical paths and thereby decreases the cost in \eqref{eq:defender_utility_edge}, which yields a contradiction. This shows that our claim is true.

Now, suppose that the defender's cost function  $C_k(x_{k})$ does not have a unique minimizer. Then, there exist two different minimizers $x_k^1$ and $x_k^2$. Let $\bar{E} \subseteq \mathcal{E}$ be the set of edges where the investments are different in the two solutions. 
For each asset $v_m \in V_k$, let $\mathcal{\bar{P}}_m \subseteq \mathcal{P}_m$ be the set of all paths from the source to $v_m$ that pass through at least one edge in $\bar{E}$. Define $x_k^3 = \frac{1}{2}(x_k^1 + x_k^2)$, which must also be an optimal solution of $C_k(x_k)$ (by convexity of $C_k(x_{k})$, as established in Lemma~\ref{lemma:edge_convex}). Furthermore, a component of $x_k^3$ is nonzero whenever at least one of the corresponding components in $x_k^1$ or $x_k^2$ is nonzero.  In particular, $x_k^3$ is nonzero on each edge in $\bar{E}$.  

For any investment vector $x_k$, given a path $P$, we use $x_{k,P}$ to denote the vector of investments on edges on the path $P$. For each asset $v_m \in V_k$ and path $P \in \mathcal{P}_m$, denote $ h_{P}(x_{k,P}) \triangleq \displaystyle \sum_{(v_i, v_j) \in P} (-\log(p_{i,j}(x_{i,j})))^{\alpha_k}$. By Lemma~\ref{lemma:perceived_convex}, each term of the form $(-\log(p_{i,j}(x_{i,j})))^{\alpha_k}$ is strictly concave in $x_{i,j}$ when $\alpha_k \in (0,1)$.  Thus,  $h_{P}(x_{k,P})$ is strictly concave in $x_{k,P}$ for $\alpha_k \in (0,1)$.

Then, using  \eqref{eq:prelec_k}, the value of the path $P$ is given by 
\begin{equation*}
f_P(x_{k,P}) \triangleq
\prod_{(v_i,v_j) \in P}w_k(p_{i,j}(x_{i,j})) = \exp(-h_P(x_{k,P})).
\end{equation*}
Note that by strict concavity of $h_{P}(x_{k,P})$ in $x_{k,P}$ when $\alpha_k \in (0,1)$,  $f_P(x_{k,P})$ is strictly convex in $x_{k,P}$ when $\alpha_k \in (0,1)$.
For each asset $v_m \in V_k$, the value of each critical path is
\begin{align*}
g_m(x_k) &\triangleq \underset{P \in \mathcal{P}_{m}} {\text{max}} f_P(x_{k,P}) \\
& = {\text{max}} \left(\underset{P \in \mathcal{\bar{P}}_{m}} {\text{max}} f_P(x_{k,P}), \underset{P \in \mathcal{P}_{m} \setminus \mathcal{\bar{P}}_{m}} {\text{max}} f_P(x_{k,P})\right).
\end{align*}
Now, returning to the optimal investment vector $x_k^3$, define 
\begin{equation*}
\hat{M} \triangleq \{v_m \in V_k | \underset{P \in \mathcal{\bar{P}}_{m}} {\text{max}} f_P(x_{k,P}^3) \geq 
\underset{P \in \mathcal{P}_{m} \setminus \mathcal{\bar{P}}_{m}} {\text{max}} f_P(x_{k,P}^3) \}.
\end{equation*}
In other words, $\hat{M}$ is the set of assets for which there is a critical path (under the investment vector $x_k^3$) that passes through the set $\bar{E}$ (where the optimal investments $x_k^1$ and $x_{k}^2$ differ).  Now there are two cases. The first case is when $\hat{M}$ is nonempty.
We have (from \eqref{eq:defender_utility_edge})
\begin{align*}
& C_k(x_{k}^3) = \sum_{v_m \notin \hat{M}} L_m \hspace{1mm} g_m(x_{k}^3) + \sum_{v_m \in \hat{M}} L_m \hspace{1mm} g_m(x_{k}^3) && \\ 
& \quad \stackrel{(a)}{=} \sum_{v_m \notin \hat{M}} L_m \hspace{1mm} \underset{P \in \mathcal{P}_{m} \setminus \mathcal{\bar{P}}_{m}} {\text{max}} f_P(x_{k,P}^3)  + \sum_{v_m \in \hat{M}} L_m \hspace{1mm}  \underset{P \in \mathcal{\bar{P}}_{m}} {\text{max}} f_P(x_{k,P}^3) && \\
& \quad \stackrel{(b)}{<} \sum_{v_m \notin \hat{M}} L_m \hspace{1mm} \frac{1}{2} \underset{P \in \mathcal{P}_{m} \setminus \mathcal{\bar{P}}_{m}} {\text{max}} (f_P(x_{k,P}^1) + f_P(x_{k,P}^2)) + \sum_{v_m \in \hat{M}} L_m \hspace{1mm} \frac{1}{2} \underset{P \in \mathcal{\bar{P}}_{m}} {\text{max}} (f_P(x_{k,P}^1) + f_P(x_{k,P}^2)) && \\
& \quad \stackrel{(c)}{\leq} \sum_{v_m \notin \hat{M}} L_m \hspace{1mm} \frac{1}{2} \underset{P \in \mathcal{P}_{m}} {\text{max}} (f_P(x_{k,P}^1) + f_P(x_{k,P}^2))  + \sum_{v_m \in \hat{M}} L_m \hspace{1mm} \frac{1}{2} \underset{P \in \mathcal{P}_{m}} {\text{max}} (f_P(x_{k,P}^1) + f_P(x_{k,P}^2)) && \\ 
& \quad \hspace{5mm} && \\
& \quad \stackrel{(d)}{\leq} \frac{1}{2} \sum_{v_m \notin \hat{M}} L_m   \left(\underset{P \in \mathcal{P}_{m}} {\text{max}} f_P(x_{k,P}^1) +    \underset{P \in \mathcal{P}_{m}} {\text{max}} f_P(x_{k,P}^2)\right) + \frac{1}{2} \sum_{v_m \in \hat{M}} L_m  \left(\underset{P \in \mathcal{P}_{m}} {\text{max}} f_P(x_{k,P}^1) +    \underset{P \in \mathcal{P}_{m}} {\text{max}} f_P(x_{k,P}^2)\right) &&\\
& \quad = \frac{1}{2} \left( \sum_{v_m \in V_{k}} L_m \hspace{1mm} g_m(x_{k}^1) + \sum_{v_m \in V_{k}} L_m \hspace{1mm} g_m(x_{k}^2) \right). &&
\end{align*}
Note that (a) holds from the definition of $\hat{M}$. Also, (b) holds since for each $P \in \mathcal{\bar{P}}_{m}, f_P(x_{k,P}^3) <  \frac{1}{2}(f_P(x_{k,P}^1) + f_P(x_{k,P}^2))$ by strict convexity of $f_P$ in $x_{k,P}$ and since $x_{k,P}^3$ is a strict convex combination of $x_{k,P}^1$ and $x_{k,P}^2$ (by definition of $\mathcal{\bar{P}}_m$). Thus, for $v_m \in \hat{M}$, $\underset{P \in \mathcal{\bar{P}}_{m}} {\text{max}} f_P(x_{k,P}^3) < \underset{P \in \mathcal{\bar{P}}_{m}} {\text{max}} \frac{1}{2}(f_P(x_{k,P}^1) + f_P(x_{k,P}^2))$. Further, (c) holds since the maximum over a  subset of the  paths ($\mathcal{\bar{P}}_{m}$ or $\mathcal{P}_m \setminus \mathcal{\bar{P}}_{m}$) is less than or equal the maximum over the set of all paths $\mathcal{P}_{m}$. Finally, (d) holds as the maximum of a sum of elements is at most the sum of maxima. Thus, $C_k(x_{k}^3) <  \frac{1}{2} (C_k(x_{k}^1) + C_k(x_{k}^2))$ which yields a contradiction to the optimality of $x^1_k$ and $x^2_k$.

In the second case, suppose $\hat{M}$ is  empty. Thus, $\forall v_m \in V_k$, $\underset{P \in \mathcal{\bar{P}}_{m}} {\text{max}} f_P(x_{k,P}^3) < \underset{P \in \mathcal{P}_{m} \setminus \mathcal{\bar{P}}_{m}} {\text{max}} f_P(x_{k,P}^3)$. In other words, for all assets $v_m \in V_k$, no critical paths go through the edge set $\bar{E}$ (since $\bar{\mathcal{P}}_m$ contains all such paths).     However, $x_k^3$ has nonzero investments on edges in $\bar{E}$.   Thus, $x_k^3$ cannot be an optimal solution (by the claim at the start of the proof). Thus, the second case is also not possible.  Hence there cannot be two different optimal solutions, and therefore the optimal investments for the defender $D_k$ are unique.
\end{proof}


In contrast to the above result, the optimal investments by a non-behavioral defender (i.e., $\alpha=1$) need not be unique. To see this, consider an attack graph where the probability of successful attack on each edge is given by the exponential function \eqref{eq:expon_prob_func}.  As argued in equation \eqref{eq:path_probability}, the probability of successful attack on any given path is a function of the {\it sum} of the security investments on {\it all} the edges in that path. Thus, given an optimal set of investments by a non-behavioral defender, any other set of investments that maintains the same total investment on each path of the graph is also optimal.

\subsection{Locations of Optimal Investments for Behavioral and Non-Behavioral Defenders}

We next study differences in the {\it locations} of the optimal investments by behavioral and non-behavioral defenders.  In particular, we first characterize the optimal investments by a non-behavioral defender who is protecting a single asset, and subsequently compare that to the investments made by a behavioral defender.  In the following result, we use the notion of a {\it min-cut} in the graph.  Specifically, given two nodes $s$ and $t$ in the graph, an edge-cut is a set of edges $\mathcal{E}_{c} \subset \mathcal{E}$ such that removing $\mathcal{E}_c$ from the graph also removes all paths from $s$ to $t$.  A min-cut is an edge-cut of smallest cardinality over all possible edge-cuts \cite{west2001introduction}.

\begin{proposition}\label{prop: non-behavioral min cut}
Consider an attack graph $ G = (V, \mathcal{E})$. Let the attack success probability under security investments be given by $ p_{i,j}(x_{i,j}) = e^{-x_{i,j}} $, where  $ x_{i,j} \in \mathbb{R}_{\ge 0}$ is the investment  on edge $(v_{i},v_{j}) $.  Suppose there is a single target asset $v_t$ (i.e., all other assets have loss $0$). Let $ \mathcal{E}_{c} \subseteq \mathcal{E} $ be a min-cut between the source node $ v_s $ and the target $ v_t $. Then, it is optimal for a non-behavioral defender $ D_k $ to distribute all her investments equally only on the edge set $ \mathcal{E}_{c} $ in order to minimize \eqref{eq:defender_utility}.
\end{proposition}

\begin{proof}
Let $ N = |\mathcal{E}_c|$ represent the number of edges in the min-cut set $ \mathcal{E}_{c} $. Let $ B $ be the defender's budget. 

Consider any optimal investment of that budget. Recall from \eqref{eq:path_probability} that for probability functions of the form \eqref{eq:expon_prob_func}, the probability of a successful attack of the target along a certain path $P$ is a decreasing function of the sum of the investments on the edges on that path.  Using Menger's theorem \cite{west2001introduction}, there are $ N $ edge-disjoint paths between $ v_s $ and $v_t$ in $ G $. At least one of those paths has total investment at most $\frac{B}{N}$.  Therefore, the path with highest probability of attack from $ v_s $ to $v_t$ has total investment at most $ \frac{B}{N} $.

Now consider investing $ \frac{B}{N} $ on each edge in the min-cut. Since every path from $ v_s $ to $v_t$ goes through at least one edge in $ \mathcal{E}_{c} $, every path has at least $\frac{B}{N}$ in total investment. Thus, it is optimal to only invest on edges in $\mathcal{E}_{c}$.

Finally, consider investing non-equally on edges in $\mathcal{E}_c$ where an edge $(v_i, v_j) \in \mathcal{E}_{c}$ has investment $x_{i,j} <  \frac{B}{N}$. Under this investment, since there are $N$ edge-disjoint paths from $v_s$ to $v_t$ in $G$, there exists a path $ P $ from $v_s$ to $v_t$ that has total investment less than $\frac{B}{N}$. Thus, the path with the highest probability of attack has a probability of attack larger than $\exp(-\frac{B}{N})$ (which would be obtained when investing $\frac{B}{N}$ equally on each edge in $\mathcal{E}_c$). Therefore, the true expected cost in \eqref{eq:defender_utility} is higher with this non-equal investment. Thus, the optimal investment on $\mathcal{E}_c$ must contain $\frac{B}{N}$ investment on each edge in $\mathcal{E}_c$.
\end{proof}

\begin{remark}
\normalfont The above result will continue to hold for more general probability functions $p_{m,n}(x_{m,n}) = p_{m,n}^0 e^{-x_{m,n}}$ with $p_{m,n}^0 \neq 1$ if $\displaystyle {\prod_{(v_{m},v_{n}) \in P} p_{m,n}^0}$ is the same for every path $P \in \mathcal{P}_{t}$. The baseline successful attack probability is then the same along every path to $v_{t}$, and thus optimal investments can be restricted to the edges in the min-cut set. $\hfill \blacksquare$
\end{remark}
The conclusion of Proposition~\ref{prop: non-behavioral min cut} no longer holds when we consider the investments by a behavioral defender (i.e., with $\alpha_k < 1$), as illustrated by the following example.

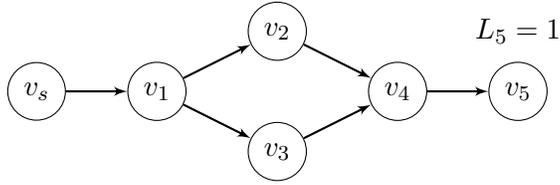
\begin{figure}
\centering
\begin{tikzpicture}[scale=0.8]

\tikzset{edge/.style = {->,> = latex'}}

\node[draw,shape=circle] (vs) at (-2,0) {$v_s$};
\node[draw,shape=circle] (v1) at (0,0) {$v_1$};
\node[draw,shape=circle] (v2) at (2,1) {$v_2$};
\node[draw,shape=circle] (v3) at (2,-1) {$v_3$};
\node[draw,shape=circle] (v4) at (4,0) {$v_4$};
\node[draw,shape=circle] (v5) at (6,0) {$v_5$};

\node[shape=circle] (L1) at (6,1) {$L_5=1$};

\draw[edge,thick] (vs) to (v1);
\draw[edge,thick] (v1) to (v2);
\draw[edge,thick] (v1) to (v3);
\draw[edge,thick] (v2) to (v4);
\draw[edge,thick] (v3) to (v4);
\draw[edge,thick] (v4) to (v5);

\end{tikzpicture}
\caption{ An attack graph where a behavioral defender makes suboptimal investment decisions.}
\label{fig:split_join}
\end{figure}

\begin{example}
Consider the attack graph shown in Figure~\ref{fig:split_join}, with a single defender $D$ (we will drop the subscript $k$ for ease of notation in this example) and a single target asset $v_5$ with a loss of $L_5 = 1$ if successfully attacked.  Let the defender's budget be $B$, and let the probability of successful attack on each edge $(v_i, v_j)$ be given by $p_{i,j}(x_{i,j}) = e^{-x_{i,j}}$,
where $x_{i,j}$ is the investment on that edge.

This graph has two possible min-cuts, both of size $1$: the edge $(v_s, v_1)$, and the edge $(v_4, v_5)$.  Thus, by Proposition~\ref{prop: non-behavioral min cut}, it is optimal for a non-behavioral defender to put all of her budget on either one of these edges.  

Now consider a behavioral defender with $\alpha < 1$.  With the above expression for $p_{i,j}(x_{i,j})$ and using the Prelec function \eqref{eq:prelec_k}, we have
$w(p_{i,j}(x_{i,j})) = e^{-x_{i,j}^{\alpha}}$. 
Thus, the perceived expected cost function \eqref{eq:defender_utility_edge}  is given by
\begin{equation*}
C(\mathbf{x}) =  \max \left(e^{-x_{s,1}^{\alpha} - x_{1,2}^{\alpha} - x_{2,4}^{\alpha} - x_{4,5}^{\alpha}},e^{-x_{s,1}^{\alpha} - x_{1,3}^{\alpha} - x_{3,4}^{\alpha} - x_{4,5}^{\alpha}}\right),
\end{equation*}
corresponding to the two paths from the source $v_s$ to the target $v_t$.  One can verify (using the KKT conditions) that the optimal investments are given by 
\begin{equation}\label{eq: behavioral optimal investments}
\begin{aligned}
x_{1,2} &= x_{2,4} = x_{1,3} = x_{3,4} = 2^{\frac{1}{\alpha-1}} x_{s,1} ~,\\
x_{4,5} &= x_{s,1} = \frac{B-4x_{1,2}}{2}= \frac{B}{2+4( 2^{\frac{1}{\alpha-1}}) }~. 
\end{aligned}
\end{equation}
Thus, for the true expected cost function \eqref{eq:defender_utility}, the optimal investments (corresponding to  the non-behavioral defender) yield a true expected cost of $ e^{-B}$, whereas the investments of the behavioral defender yield a true expected cost of $ e^{-2^\frac{\alpha}{\alpha-1}} e^{-\frac{B}{1+ 2^{\frac{\alpha}{\alpha-1}} }}$, which is larger than that of the non-behavioral defender. %
\end{example}

The above example illustrates a key phenomenon: as the defender's perception of probabilities becomes increasingly skewed (captured by $\alpha$ becoming smaller), she shifts more of her investments from the min-cut edges to the edges on the parallel paths between $v_1$ and $v_4$.  This is in contrast to the optimal investments (made by the non-behavioral defender) which lie entirely on the min-cut edges.  Indeed, by taking the limit as $ \alpha \uparrow 1$, we have 
\begin{equation*}
x_{i,j} = \lim_{\alpha \uparrow 1} \hspace{1mm} 2^ \frac{1}{\alpha -1} \hspace{1mm} x_{s,1} = 2^{-\infty} \hspace{1mm} x_{s,1} = 0
\end{equation*}
for edges $(v_i,v_j)$ on the two parallel portions of the graph.

We now use this insight to identify graphs where the behavioral defender finds that investing only on the min-cut edges is not optimal.

\begin{proposition}\label{prop:behavioral-suboptimal}
Consider an attack graph $G$ with a source $v_s$ and a target $v_t$.  Let $\mathcal{E}_{c}$ be a min-cut between $v_s$ and $v_t$, with size $|\mathcal{E}_{c}| = N$.  Suppose the graph contains another edge cut
$\mathcal{E}^{'}_{c}$ such that $\mathcal{E}^{'}_{c} \cap \mathcal{E}_{c} = \emptyset$ and $|\mathcal{E}^{'}_{c}| > |\mathcal{E}_{c}|$. 
Let the probability of successful attack on each edge $(v_i, v_j) \in \mathcal{E}$ be given by $p_{i,j}(x_{i,j}) = e^{-x_{i,j}}$, where $x_{i,j}$ is the investment on that edge.  Let $B$ be the budget of the defender.  Then, if $0 < \alpha_k < 1$, investing solely on the min-cut set $\mathcal{E}_{c}$ is not optimal from the perspective of a behavioral defender. 
\end{proposition}

\begin{proof}
Denote $M = |\mathcal{E}^{'}_{c}| > |\mathcal{E}_{c}| = N$.  By Proposition~\ref{prop: non-behavioral min cut}, it is optimal to invest the entire budget uniformly on edges in $\mathcal{E}_{c}$ in order to minimize the cost function \eqref{eq:defender_utility}.  We will show that this investment is not optimal with respect to the behavioral defender's cost function \eqref{eq:defender_utility_edge}; we will drop the subscript $k$ in $\alpha_k$ for ease of notation.

Starting with the optimal investments on the min edge cut $\mathcal{E}_{c}$ where each edge in $\mathcal{E}_{c}$ has nonzero investment (as given
by Proposition~\ref{prop: non-behavioral min cut}), remove a small  investment $\epsilon$ from each of those $N$ edges, and add an investment of $\frac{N\epsilon}{M}$ to each of the edges in $\mathcal{E}^{'}_{c}$.  
We show that when $\epsilon$ is sufficiently small, this will lead to a net reduction in perceived probability of successful attack on each path from $v_s$ to $v_t$. 

Consider any arbitrary path $P$ from $v_s$ to $v_t$. Starting with the investments only on the minimum edge cut $\mathcal{E}_{c}$, the perceived probability of successful attack on path $P$ will be 
\begin{equation*}
f_{1}(\mathbf{x}) \triangleq  \exp\Big(- \sum_{\substack{(v_i,v_j) \in \mathcal{E}_{c},\\(v_i,v_j) \in P}} x_{i,j}^{\alpha}\Big). 
\end{equation*}
After removing $\epsilon$ investment from each of the $N$ edges in $\mathcal{E}_{c}$, and adding an investment of $\frac{N\epsilon}{M}$ to each of the edges in $\mathcal{E}^{'}_{c}$, the perceived probability on path $P$ will be: 
\begin{equation*}
f_{2}(\mathbf{x}) \triangleq \exp\Big(-\sum_{\substack{(v_i,v_j) \in \mathcal{E}^{'}_{c}, \\ (v_i,v_j) \in P}}  \Big(\frac{N\epsilon}{M}\Big)^{\alpha} - \sum_{\substack{(v_i,v_j) \in \mathcal{E}_{c}, \\ (v_i,v_j) \in P}} (x_{i,j} - \epsilon)^{\alpha}  \Big). 
\end{equation*}
The net reduction in perceived probability on path $P$ will be positive if 
$f_{2}(\mathbf{x}) < f_{1}(\mathbf{x})$, i.e., 
\begin{equation}
\sum_{\substack{(v_i,v_j) \in \mathcal{E}^{'}_{c},\\ (v_i,v_j) \in P }}  \left(\frac{N\epsilon}{M}\right)^{\alpha} + \sum_{\substack{(v_i,v_j) \in \mathcal{E}_{c}, \\ (v_i,v_j) \in P}} (x_{i,j} - \epsilon)^{\alpha} > \sum_{\substack{(v_i,v_j) \in \mathcal{E}_{c},\\ (v_i,v_j) \in P}} x_{i,j}^{\alpha}.
\label{eq:f2_lt_f1}
\end{equation}
If we define
$$
f(\epsilon) \triangleq \sum_{\substack{(v_i,v_j) \in \mathcal{E}^{'}_{c}, \\ (v_i,v_j) \in P}}  \Big(\frac{N\epsilon}{M}\Big)^{\alpha} + \sum_{\substack{(v_i,v_j) \in \mathcal{E}_{c},\\ (v_i,v_j) \in P}} (x_{i,j} - \epsilon)^{\alpha},
$$
we see that inequality \eqref{eq:f2_lt_f1} is equivalent to showing that $f(\epsilon) > f(0)$.  We have
$$
\frac{df}{d\epsilon} = \frac{\alpha N}{M} \sum_{\substack{(v_i,v_j) \in \mathcal{E}^{'}_{c},\\ (v_i,v_j) \in P}} \Big(\frac{N\epsilon}{M}\Big)^{\alpha-1}  - \alpha \sum_{\substack{(v_i,v_j) \in \mathcal{E}_{c}, \\ (v_i,v_j) \in P}} (x_{i,j} - \epsilon)^{\alpha - 1}.
$$
Note that $\lim_{\epsilon \downarrow 0} \frac{df}{d\epsilon} = \infty$ which shows that $f(\epsilon)$ is increasing in $\epsilon$ for sufficiently small $\epsilon$. Therefore, $f_{2}(\mathbf{x}) < f_{1}(\mathbf{x})$ for sufficiently small $\epsilon$. Since this analysis holds for every path from $v_s$ to $v_t$, this investment profile outperforms investing purely on the minimum edge cut. 
\end{proof}

Note that the graph in Figure~\ref{fig:split_join} satisfies the conditions in the above result, with $\mathcal{E}_{c} = {(v_4,v_5)}$, $\mathcal{E}^{'}_{c} = \{(v_1,v_2), (v_1,v_3)\}$. 

Having established properties of the optimal investment decisions for behavioral and non-behavioral defenders, we next turn our attention to the Behavioral Security Game with multiple defenders, introduced in Section~\ref{sec:behavioralclasses}.

%% file: OptimalInvestmentsPNE.tex
\section{Analysis of Multi-Defender Games}\label{sec: PNE}

\begin{figure*}[t!] 
\centering
\begin{subfigure}[t]{.46\textwidth}
\centering
  \begin{tikzpicture}[scale=1]

\tikzset{edge/.style = {->,> = latex'}}

\node[draw,shape=circle] (x1) at (-1,1) {$v_s$};
\node[draw,shape=circle] (x2) at (1,1) {$v_1$};
\node[draw,shape=circle] (x3) at (3,1) {$v_2$};
\node[draw,shape=circle] (x4) at (-1,-1) {$v_3$};
\node[draw,shape=circle] (x5) at (1,-1) {$v_4$};
\node[draw,shape=circle] (x6) at (3,-1) {$v_5$};

\node[shape=circle] (L1) at (1,-1.5) {$L_1=1$};
\node[shape=circle] (L2) at (3,-1.5) {$L_2=1$};
\draw[edge,thick] (x1) to[edge node={node[above,xshift=-0cm]{$4$}},edge node={node[below,xshift=0cm]{$0$}}] (x2);
\draw[edge,thick] (x2) to[edge node={node[above,xshift=-0cm]{$0$}},edge node={node[below,xshift=0cm]{$4$}}] (x3);
\draw[edge,thick] (x1) to[edge node={node[left,xshift=-0cm]{$4$}},edge node={node[right,xshift=0cm]{$0$}}] (x4);
\draw[edge,thick] (x2) to[edge node={node[left,xshift=-0cm]{$4$}},edge node={node[right,xshift=0cm]{$0$}}] (x5);
\draw[edge,thick] (x3) to[edge node={node[left,xshift=-0cm]{$0$}},edge node={node[right,xshift=0cm]{$4$}}] (x6);
\draw[edge,thick] (x4) to[edge node={node[above,xshift=-0cm]{$4$}},edge node={node[below,xshift=0cm]{$0$}}] (x5);
\draw[edge,thick] (x5) to[edge node={node[above,xshift=-0cm]{$0$}},edge node={node[below,xshift=0cm]{$4$}}] (x6);
\end{tikzpicture}
\caption{}
\label{fig:MPNE_example_both_behavioral_1}
\end{subfigure}
\begin{subfigure}[t]{.46\textwidth}
\centering
 \begin{tikzpicture}[scale=1]

\tikzset{edge/.style = {->,> = latex'}}

\node[draw,shape=circle] (x1) at (-1,1) {$v_s$};
\node[draw,shape=circle] (x2) at (1,1) {$v_1$};
\node[draw,shape=circle] (x3) at (3,1) {$v_2$};
\node[draw,shape=circle] (x4) at (-1,-1) {$v_3$};
\node[draw,shape=circle] (x5) at (1,-1) {$v_4$};
\node[draw,shape=circle] (x6) at (3,-1) {$v_5$};

\node[shape=circle] (L1) at (1,-1.5) {$L_1=1$};
\node[shape=circle] (L2) at (3,-1.5) {$L_2=1$};
\draw[edge,thick] (x1) to[edge node={node[above,xshift=-0cm]{$1$}},edge node={node[below,xshift=0cm]{$4$}}] (x2);
\draw[edge,thick] (x2) to[edge node={node[above,xshift=-0cm]{$0$}},edge node={node[below,xshift=0cm]{$3.14$}}] (x3);
\draw[edge,thick] (x1) to[edge node={node[left,xshift=-0cm]{$5$}},edge node={node[right,xshift=0cm]{$0$}}] (x4);
\draw[edge,thick] (x2) to[edge node={node[left,xshift=-0cm]{$5$}},edge node={node[right,xshift=0cm]{$0$}}] (x5);
\draw[edge,thick] (x3) to[edge node={node[left,xshift=-0cm]{$0$}},edge node={node[right,xshift=0cm]{$3.14$}}] (x6);
\draw[edge,thick] (x4) to[edge node={node[above,xshift=-0cm]{$5$}},edge node={node[below,xshift=0cm]{$0$}}] (x5);
\draw[edge,thick] (x5) to[edge node={node[above,xshift=-0cm]{$0$}},edge node={node[below,xshift=0cm]{$1.72$}}] (x6);
\end{tikzpicture}
 \caption{}
 \label{fig:MPNE_example_both_behavioral_2}
 \end{subfigure} 
\caption{An instance of a Behavioral Security Game with multiple PNE. Defenders $D_{1}$ and $D_{2}$ are behavioral decision-makers with $\alpha_{1} = \alpha_{2} = 0.5$. The numbers above/left and below/right of the edges represent investments by $D_1$ and $D_2$, respectively.}
\label{fig:MPNE_main}
\end{figure*}
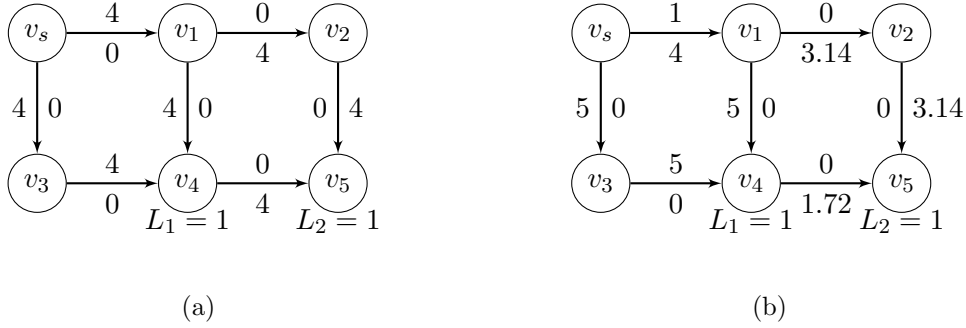


\subsection{Existence of a PNE}
We first establish the existence of a Pure Strategy Nash Equilibrium (PNE) for the class of behavioral games defined in Section~\ref{sec:behavioralclasses}. Recall that a profile of security investments by the defenders is said to be a PNE if no defender can decrease her cost by unilaterally changing her security investment.

\begin{proposition}\label{prop:existence_pne}
Under Assumption \ref{ass:attack_prob}, the Behavioral Security Game possesses a pure strategy Nash equilibrium (PNE) when $\alpha_{k} \in (0,1]$ for each defender $D_{k}$.
\end{proposition}

\begin{proof}
The feasible defense strategy space $ X_{k} $ in \eqref{eq:defense_strategy_space} is nonempty, compact and convex for each defender $ D_{k} $.  Furthermore, for all $D_k \in \mathcal{D}$ and investment vectors $\mathbf{x}_{-k}$, the cost function $C(x_k, \mathbf{x}_{-k})$ in \eqref{eq:defender_utility_edge} is convex in $x_k \in X_k$; this follows from Lemma~\ref{lemma:edge_convex} and the fact that the investment $x_{i,j}$ on each edge is a sum of the investments of all players on that edge.  As a result, the Behavioral Security Game is an instance of {\it concave games}, which always have a PNE \cite{rosen1965existence}.
\end{proof}

Note that in contrast to the best responses by each player (which were unique when $\alpha_k \in (0,1)$, as shown in Proposition~\ref{prop: uniqueness of investments}), the PNE of Behavioral Security Games is not unique in general.  We illustrate this through the following example.

\begin{example}\label{ex:MPNE_both_behavioral}
Consider the attack graph of Figure \ref{fig:MPNE_main}. There are two defenders, $D_1$ and $D_2$,  where
defender $D_1$ wishes to protect node $v_4$, and defender $D_2$ wishes to protect node $v_5$. Suppose that $D_1$ has a budget $B_{1} = 16$ and $D_{2}$ has $B_{2} = 12$. Figs.  \ref{fig:MPNE_example_both_behavioral_1} and \ref{fig:MPNE_example_both_behavioral_2} illustrate two distinct PNE for this game. We obtained multiple Nash equilibria by varying the starting investment decision of defender $D_{1}$ and  then following best response dynamics until the investments converged to an equilibrium.

It is interesting to note that these two Nash equilibria lead to different costs for the defenders. First, for the Nash equilibrium of Figure \ref{fig:MPNE_example_both_behavioral_1}, defender $D_{1}$'s  perceived expected cost, given by \eqref{eq:defender_utility_edge}, is equal to $\exp(-4)$, while her true expected cost, given by \eqref{eq:defender_utility}, is equal to  $\exp(-8)$. Defender $D_{2}$ has a perceived expected cost of $\exp(-6)$, and a true expected cost of $\exp(-12)$.  
In contrast, for the Nash equilibrium in Figure \ref{fig:MPNE_example_both_behavioral_2}, defender $D_{1}$ has a perceived expected cost of $\exp(-2 \enspace \sqrt[]{5})$ and a true expected cost of $\exp(-10)$. Defender $D_{2}$ has a perceived expected cost of $\exp(-5.78)$ and a true expected cost of $\exp(-11.28)$. 

As a result, the equilibrium in Figure \ref{fig:MPNE_example_both_behavioral_1} is preferred by defender $D_{2}$,  while the equilibrium in Figure \ref{fig:MPNE_example_both_behavioral_2} has a lower expected cost (both perceived and real) for defender $D_{1}$. Note also that the total expected cost (i.e., sum of the true expected costs of defenders $D_{1}$ and $D_{2}$) is lower in the equilibrium in Figure \ref{fig:MPNE_example_both_behavioral_2}; that is, the PNE of Figure \ref{fig:MPNE_example_both_behavioral_2} would be preferred from a social planner's perspective.
\end{example}

\input{Security_under_Anarchy}

\input{Optimal_investment_multiple_defenders}

%% file: Security_under_Anarchy.tex
\subsection{Measuring the Inefficiency of PNE: The Price of Behavioral Anarchy}\label{sec:PoBA}

\begin{figure}
\centering
\begin{tikzpicture}[scale=0.8]

\tikzset{edge/.style = {->,> = latex'}}

\node[draw,shape=circle] (vs) at (-2,0) {$v_s$};
\node[draw,shape=circle] (v1) at (0,0) {$v_1$};
\node[draw,shape=circle] (v2) at (2,0) {$v_2$};
\node[draw,shape=circle] (v3) at (4,0) {$v_3$};
\node[draw,shape=circle] (vk) at (6,0) {$v_{K}$};

\node[shape=circle] (L1) at (0,1) {$L_1$};
\node[shape=circle] (L2) at (2,1) {$L_2$};
\node[shape=circle] (L3) at (4,1) {$L_3$};
\node[shape=circle] (LK) at (6,1) {$L_{K}$};

\draw[edge,thick] (vs) to (v1);
\draw[edge,thick] (v1) to (v2);
\draw[edge,thick] (v2) to (v3);
\draw[edge,thick,dotted] (v3) to (vk);
\end{tikzpicture}
\caption{An attack graph where PoBA is lower bounded by $(1-\epsilon) \exp(B)$ which shows that the upper bound obtained in Proposition~\ref{prop:POBA_upper_bound} is tight.}
\label{fig:series_SOA}
\end{figure}
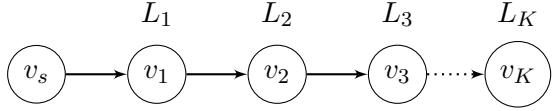

The notion of Price of Anarchy (PoA) is often used to quantify the inefficiency of Nash equilibrium 
 compared to the socially optimal outcome \cite{roughgarden2003price}. Specifically, the Price of Anarchy is defined as the ratio of the highest total system true expected cost at a PNE to the total system true expected cost at the social optimum. 
 For our setting, we seek to define a measure to capture the inefficiencies of the equilibrium due to both the defenders' individual strategic behavior and their behavioral decision-making. We thus define the Price of Behavioral Anarchy (PoBA) as the ratio of total system true expected cost of behavioral defenders at the worst PNE (i.e., the PNE with the largest total true expected cost over all PNE),
to the total system true expected cost at the social optimum (computed by a non-behavioral social planner).\footnote{One could also consider the impact of a behavioral social planner; since the goal of our paper is to quantify the (objective) inefficiencies due to behavioral decision making, we leave the study of behavioral social planner for future work.} 

Specifically, we define $\hat{C}(\mathbf{x}) \triangleq \sum_{D_k \in \mathcal{D}} \hat{C}_k(\mathbf{x})$, where $\hat{C}_k$ (defined in \eqref{eq:defender_utility}) is the true expected cost faced by defender $D_k$ under the investment vector $\mathbf{x}$. Let $X^{\mathtt{NE}}:=\{\mathbf{\bar{x}} \in \mathbb{R}^{|\DD||\EE|}_{\geq 0}| \bar{x}_k \in \displaystyle \argmin_{x \in X_k} C_k(x,\mathbf{\bar{x}}_{-k}),$ 
$\forall D_k \in \DD \}$, i.e., $X^{\mathtt{NE}}$ is the set of all investments that constitute a PNE. We now define the Price of Behavioral Anarchy as
\begin{equation}\label{eq: Security Under Anarchy}
    PoBA = \frac{\sup_{\mathbf{\bar{x}}\in X^\mathtt{NE}} \hat{C}(\mathbf{\bar {x}})}{\hat{C}(\mathbf{x^{*}})},
\end{equation}
where $\mathbf{x^{*}}$ denotes the investments at the social optimum (computed by a non-behavioral social planner with access to the sum of all defenders' budgets). Mathematically, let $X^{\mathtt{Soc}}:=\{\mathbf{x^{*}} \in \mathbb{R}^{|\DD||\EE|}_{\geq 0}| \mathbf{1}^T \mathbf{x^{*}} \leq \sum_{\forall D_k \in \DD} B_k\}$, i.e., $X^{\mathtt{Soc}}$ is the set of all investments by the social planner and $\mathbf{x^{*}} \in \argmin_{\mathbf{x^{*}} \in X^{\mathtt{Soc}}} \hat{C}(\mathbf{x^{*}})$. When $\bar{\mathbf{x}}$ is any PNE, but not necessarily the one with the worst social cost, we refer to the ratio of $\hat{C}(\mathbf{\bar {x}})$ and $\hat{C}(\mathbf{x^{*}})$ as the ``inefficiency'' of the equilibrium. We emphasize that the costs in both the numerator and the denominator are the sum of the \emph{true} (rather than perceived) expected costs of the defenders.

We will establish upper and lower bounds on the PoBA. We first show that the PoBA is bounded if the total budget is bounded (regardless of the defenders' behavioral levels).

\begin{figure*}[t] 
\centering
\begin{subfigure}[t]{.48\textwidth}
\centering
\begin{tikzpicture}[scale=0.75]

\tikzset{edge/.style = {->,> = latex'}}

\node[draw,shape=circle] (N1) at (-2,0) {$v_s$};
\node[draw,shape=circle] (N2) at (0,1.5) {$v_1$};
\node[draw,shape=circle] (N3) at (0,-1.5) {$v_2$};
\node[draw,shape=circle] (N4) at (2,0) {$v_3$};
\node[draw,shape=circle] (N5) at (5,0) {$v_4$};

\node[shape=circle] (L1) at (2.5,-1) {$L_1=200$};
\node[shape=circle] (L2) at (5.25,1) {$L_2=200$};

\draw[edge,thick] (N1) to[edge node={node[above,xshift=-0.25cm]{$1.25$}},edge node={node[below,xshift=0.1cm]{$0$}}] (N2);
\draw[edge,thick] (N1) to[edge node={node[above,xshift=0.15cm]{$1.25$}},edge node={node[below,xshift=0cm]{$0$}}] (N3);
\draw[edge,thick] (N2) to[edge node={node[above,xshift=0.25cm]{$1.25$}},edge node={node[below,xshift=-0.1cm]{$0$}}] (N4);
\draw[edge,thick] (N3) to[edge node={node[above,xshift=-0.25cm]{$1.25$}},edge node={node[below,xshift=0.cm]{$0$}}] (N4);
\draw[edge,thick] (N4) to[edge node={node[above,xshift=-0cm]{$0$}},edge node={node[below,xshift=0cm]{$20$}}] (N5);

\end{tikzpicture}
\caption{$\alpha_1 = 1$, $\alpha_2 = 1$}
\label{fig:Non_behavioral}
\end{subfigure}
\begin{subfigure}[t]{.48\textwidth}
\centering
\begin{tikzpicture}[scale=0.75]

\tikzset{edge/.style = {->,> = latex'}}

\node[draw,shape=circle] (N1) at (-2,0) {$v_s$};
\node[draw,shape=circle] (N2) at (0,1.5) {$v_1$};
\node[draw,shape=circle] (N3) at (0,-1.5) {$v_2$};
\node[draw,shape=circle] (N4) at (2,0) {$v_3$};
\node[draw,shape=circle] (N5) at (5,0) {$v_4$};

\node[shape=circle] (L1) at (2.5,-1) {$L_1=200$};
\node[shape=circle] (L2) at (5.25,1) {$L_2=200$};

\draw[edge,thick] (N1) to[edge node={node[above,xshift=-0.25cm]{$1.25$}},edge node={node[below,xshift=0.12cm]{$1.34$}}] (N2);
\draw[edge,thick] (N1) to[edge node={node[above,xshift=0.15cm]{$1.25$}},edge node={node[below,xshift=-0.1cm,yshift=-0.1cm]{$1.34$}}] (N3);
\draw[edge,thick] (N2) to[edge node={node[above,xshift=0.25cm]{$1.25$}},edge node={node[below,xshift=-0.15cm]{$1.34$}}] (N4);
\draw[edge,thick] (N3) to[edge node={node[above,xshift=-0.25cm]{$1.25$}},edge node={node[below,yshift=-0.1cm]{$1.34$}}] (N4);
\draw[edge,thick] (N4) to[edge node={node[above,xshift=-0cm]{$0$}},edge node={node[below,xshift=0cm]{$14.64$}}] (N5);
\end{tikzpicture}
 \caption{$\alpha_{1} = 1 , \alpha_{2} = 0.6$}
 \label{fig: behavioral_beneficial}
 \end{subfigure} 
\caption{The numbers above (below) each edge represent investments by defender $D_1$ ($D_2$). In (a), the non-behavioral defender $D_1$ does not receive any investment contributions from the non-behavioral defender $D_2$. In (b), the non-behavioral defender $D_1$ benefits from the investment contributions of the behavioral defender $D_2$.}
\label{fig:Beneficial_main}
\end{figure*}
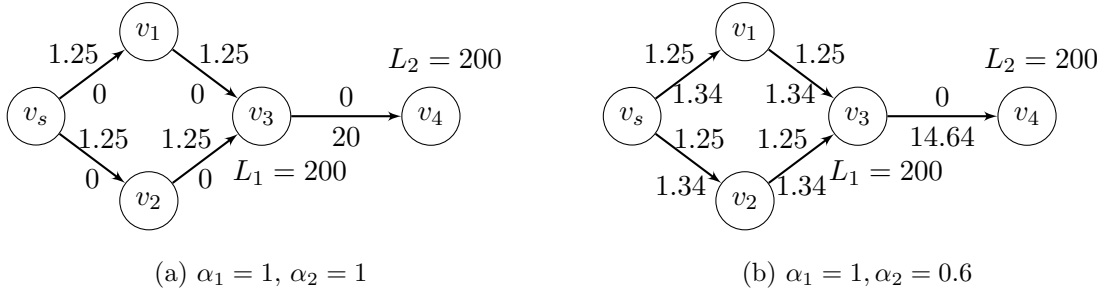

\begin{proposition}\label{prop:POBA_upper_bound}
Let the sum of the budgets available to all defenders be $B$, and let the probability of successful attack on each edge $(v_i,v_j)\in \mathcal{E}$ be  given by $p_{i,j}(x_{i,j})=e^{-x_{i,j}}$. Then, for any attack graph and any profile of behavioral levels $\{\alpha_k\}$,  $PoBA\leq \exp(B)$. 
\end{proposition}
\begin{proof}
We start with the numerator of the PoBA in \eqref{eq: Security Under Anarchy} (the total true expected cost at the worst PNE). 
Recall that each defender $D_k$ incurs a loss $L_m$ for each compromised asset $v_m$. Thus, the worst case true expected cost under any PNE (including the worst PNE) is upper bounded by $\displaystyle \sum_{D_{k} \in \mathcal{D}} \sum_{v_{m} \in V_{k}} L_{m} $ (i.e., the sum of losses of all assets). On the other hand, the denominator (the socially optimal true expected cost) is lower bounded by $\left(\displaystyle \sum_{D_{k} \in \mathcal{D}} \sum_{v_{m} \in V_{k}} L_{m}\right) \exp(-B)$ (which can only be achieved if every asset has all of the budget $B$, invested by a social planner, on  its attack path). Substituting these bounds into \eqref{eq: Security Under Anarchy}, we obtain $PoBA \le \exp(B)$.
\end{proof}


Next, we show that the upper bound on PoBA obtained in Proposition \ref{prop:POBA_upper_bound} is asymptotically tight.

\begin{proposition}
\label{prop:POBA_lower_bound}
For all $B > 0$ and $\epsilon > 0$, there exists an instance of the Behavioral Security Game with total budget $B$ such that the PoBA is lower bounded by $(1-\epsilon)\exp(B)$.
\end{proposition}
\begin{proof}
Consider the attack graph in Figure~\ref{fig:series_SOA}, where the probability of successful attack on each edge $(v_i, v_j)$ is given by \eqref{eq:expon_prob_func} with $p_{i,j}^0 = 1$. This graph contains $K$ defenders, and each defender $D_{k}$ is responsible for defending target node $v_{k}$. Assume the total security budget $B$ is divided equally between the $K$ defenders (i.e., each defender has  security budget  $\frac{B}{K}$). Let the first node have loss equal to $L_1 = K$, and the other $K-1$ nodes have loss $\frac{1}{K-1}$. Then, the socially optimal solution would put all the budget $B$ on the first link $(v_s,v_1)$, so that all nodes have probability of successful attack given by $\exp(-B)$. Thus, the denominator of \eqref{eq: Security Under Anarchy} is $\sum_{i = 1}^{K}L_i\exp(-B) = (K+1)\exp(-B)$. 

We now characterize a lower bound on the cost under a PNE (i.e., the numerator of \eqref{eq: Security Under Anarchy}).  Specifically, consider the investment profile where each defender $D_k$ puts their entire budget $\frac{B}{K}$ on the edge coming into their node $v_k$.  We claim that this is a PNE.  To show this, first consider defender $D_1$.  Since investments on edges other than $(v_s, v_1)$ do not affect the probability of successful attack at node $v_1$, it is optimal for defender $D_1$ to put all her investment on $(v_s, v_1)$.

Now consider defender $D_2$.  Given $D_1$'s investment on $(v_s, v_1)$, defender $D_2$ has to decide how to optimally spread her budget of $\frac{B}{K}$ over the two edges $(v_s, v_1)$ and $(v_1, v_2)$ in order to minimize her cost function \eqref{eq:defender_utility_edge}.  Thus, $D_2$'s optimization problem, given $D_1$'s investment, is 
\begin{equation}\label{eq:minmax_formula_POBA}
\begin{aligned}
& \underset{x_{s,1}^2 + x_{1,2}^2 = \frac{B}{K}}{\text{minimize}}
& & 
e^{-(\frac{B}{K}+ x_{s,1}^2)^{\alpha_2} - \left(x^{2}_{1,2}\right)^{\alpha_2}}.\\
\end{aligned}
\end{equation}
The unique optimal solution of \eqref{eq:minmax_formula_POBA} (for all $\alpha_2 \in (0,1)$) would be to put all $\frac{B}{K}$ into $x_{1,2}^2$ and zero on $x_{s,1}^2$.  This is also optimal (but not unique) when $\alpha_2 = 1$.

Continuing this analysis, we see that if defenders $D_1, D_2, \ldots, D_{k-1}$ have each invested $\frac{B}{K}$ on the edges incoming into their nodes, it is optimal for defender $D_k$ to also invest their entire budget $\frac{B}{K}$ on the incoming edge to $v_k$.  Thus, investing $\frac{B}{K}$ on each edge is a PNE.  

The numerator of the PoBA under this PNE is lower bounded by $L_1 \exp(-\frac{B}{K}) = K \exp(-\frac{B}{K})$.  Thus, the PoBA is lower bounded by
\begin{equation*}
PoBA \geq \frac{K \exp(-\frac{B}{K})}{(K+1)\exp(-B)} =  \frac{K \exp(-\frac{B}{K})}{(K+1)} \exp(B).
\end{equation*}
As the length of the chain grows, we have $\displaystyle \lim_{K \to \infty} \frac{K \exp(-\frac{B}{K})}{(K+1)} = 1$. Thus, for every $\epsilon > 0$, there exists $K$ large enough such that the PoBA in the line graph with $K$ nodes is lower bounded by $(1-\epsilon) \exp(B)$.
\end{proof}

\begin{remark}
\normalfont
The upper bound obtained in Proposition \ref{prop:POBA_upper_bound} is agnostic to the structure of the network, the number of defenders, and their degree of misperception of probabilities. In Proposition \ref{prop:POBA_lower_bound}, our result shows that the upper bound obtained in Proposition \ref{prop:POBA_upper_bound} is sharp (i.e., it cannot be reduced without additional assumptions on the game). For any particular instance of the problem, however, we can compute the inefficiency directly, which will depend on the network structure and other parameters of that instance. $\hfill \blacksquare$
\end{remark}
Before considering the case study, we will conclude this section with an example of an interesting phenomenon, where the (objectively) suboptimal investment decisions made by a behavioral defender with respect to their own assets can actually benefit the other defenders in the network.

%% file: Optimal_investment_multiple_defenders.tex


\begin{example}
We consider the attack graph of Figures \ref{fig:Non_behavioral} and \ref{fig: behavioral_beneficial} with two defenders, $D_1$ and $D_2$.  Defender $D_1$ wishes to protect node $v_3$, and defender $D_2$ wishes to protect node $v_4$.  Note that $D_1$'s asset  ($v_3$) is directly on the attack path to $D_2$'s asset ($v_4$).
Suppose that defender $D_1$ has a budget $B_{1} = 5$, while defender $D_2$ has a budget $ B_{2} = 20$.  The optimal investments in the following scenarios were calculated using CVX \cite{cvx}.

Suppose both defenders are non-behavioral. In this case, Proposition~\ref{prop: non-behavioral min cut} suggests that it is optimal for $D_2$ to put her entire budget on the min-cut, given by the edge $(v_3, v_4)$.  The corresponding PNE is shown in Figure \ref{fig:Non_behavioral}.  On the other hand, as indicated by Proposition~\ref{prop:behavioral-suboptimal}, investing solely on the min-cut is no longer optimal for a behavioral defender. Indeed, Figure \ref{fig: behavioral_beneficial} shows a PNE for the case where $D_2$ is behavioral with $\alpha_2 = 0.6$, and has spread some of her investment to the other edges in the attack graph.  Therefore,  $D_1$'s subnetwork will benefit due to the behavioral decision-making by $D_2$. 

It is also worth considering the total system true expected cost of the game at equilibrium, given by $ \hat{C}(\mathbf{\bar {x}}) = \hat{C}_{1}(\mathbf{\bar {x}}) + \hat{C}_{2}(\mathbf{\bar {x}})$ where $\mathbf{\bar {x}}$ is the investment at the PNE. 
For this example, when  both defenders are non-behavioral (i.e., $\alpha_{1} = \alpha_{2} = 1$), $\hat{C}(\mathbf{\bar {x}})  =  16.42$, while $ \hat{C}(\mathbf{\bar {x}}) = 1.13$ if defender $D_2$ is behavioral (with $\alpha_{1} = 1, \alpha_{2} = 0.6$). This considerable drop in the total true expected cost shows that the behavioral defender's contributions to the non-behavioral defender's  subnetwork may also be beneficial to the overall welfare of the network, especially under budget asymmetries or if defender $D_1$'s asset is more valuable.
\end{example}

%% file: example.tex
\section{Case Study}
\label{sec:example}

\begin{figure}
\centering
  \includegraphics[width=0.9\linewidth]{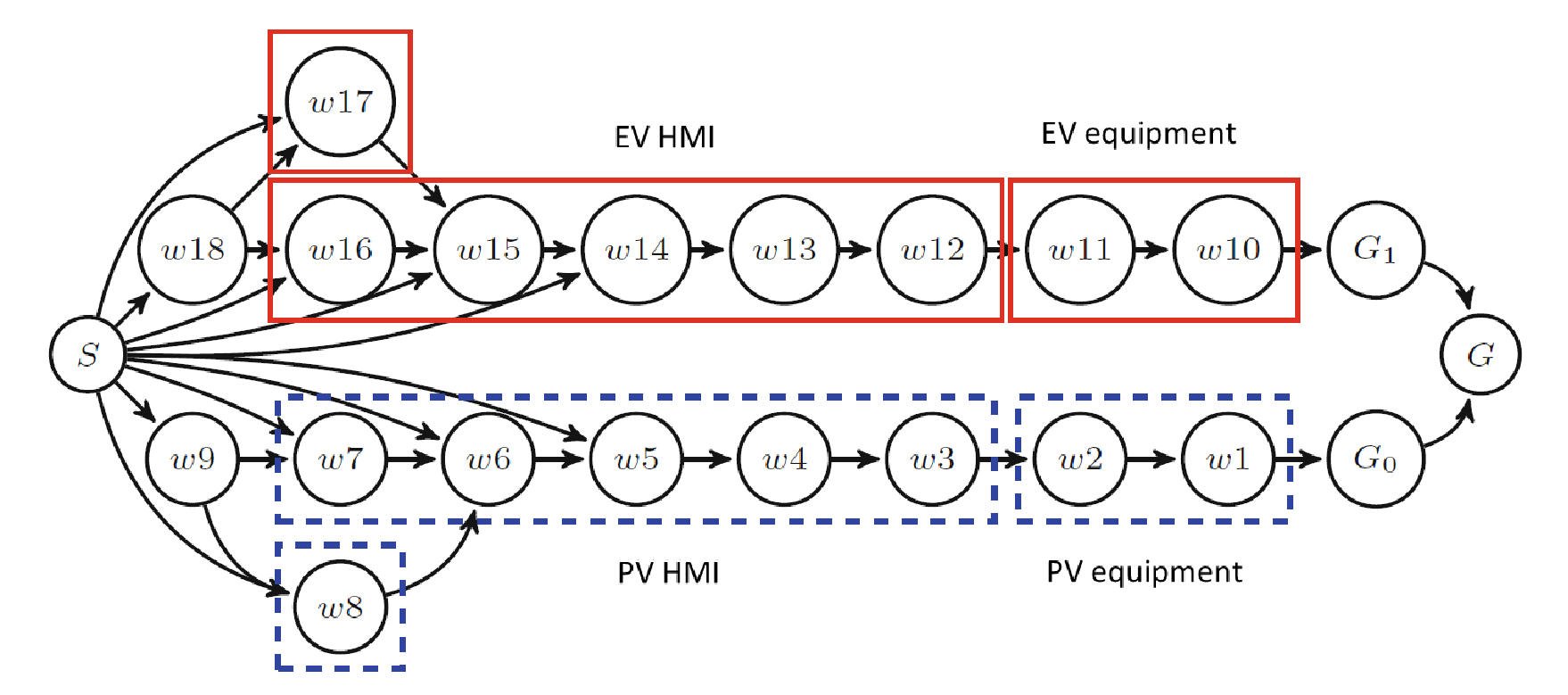}
  \caption{Attack graph of a DER.1 failure scenario adapted from
  \cite{jauhar2015model}. It shows stepping-stone attack steps that can lead to the compromise of a photovoltaic generator (PV) (i.e., $G_{0}$) or an electric vehicle charging station (EV) (i.e., $G_{1}$).} 
  \label{fig:nescor-attack-graph}
  \vspace{-0.16in}
\end{figure}

Here, we examine the outcomes of behavioral decision-making in a case study involving a distributed energy resource failure scenario, DER.1, identified by the US National Electric Sector Cybersecurity Organization Resource (NESCOR) \cite{jauhar2015model}. Figure~\ref{fig:nescor-attack-graph} is replicated from the attack graph for the DER.1 (Figure 4 in \cite{jauhar2015model}). Suppose the probability of successful attack on each edge is $p_{i,j}(x_{i,j})= e^{-x_{i,j}}$. There are two defenders, $D_1$ and $D_2$.  Defender $D_1$'s critical assets are $G_0$ and $G$, with losses of $L_0=200$ and $L=100$, respectively.  Defender $D_2$'s critical assets are $G_1$ and $G$, also with losses of $L_1=200$ and $L=100$, respectively. Note that $G$ is a shared asset among the two defenders. 

We assume that each defender has a security budget of $\frac{B}{2}$ (i.e., the budget distribution is symmetric between the two defenders). For a fair comparison, the social planner has total budget $B$. In our experiments, we use best response dynamics to find a Nash equilibrium $\bar{\mathbf{x}}$.  We then compute the socially optimal investment $\mathbf{x}^*$, and calculate the ratio given by \eqref{eq: Security Under Anarchy} to measure the inefficiency of the corresponding equilibrium.  

Figure~\ref{fig:PoBA} shows the value of this ratio as we sweep $\alpha$ (taken to be the same for both defenders) from $0$ (most behavioral) to $1$ (non-behavioral), for different values of the total budget $B$.  As the figure shows, the inefficiency of the equilibrium decreases to $1$
as $\alpha$ increases, reflecting the fact that the investment decisions become better as the defenders become less behavioral; see Section \ref{sec:single_optimal_inves_dec}. 
Furthermore,  Figure~\ref{fig:PoBA} shows that the inefficiency due to behavioral decision-making becomes exacerbated as the total budget $B$ increases.  This happens as behavioral defenders shift higher amounts of their budget to the parallel edges in the networks (i.e., not in the min-cut edge set), as suggested by Proposition~\ref{prop:behavioral-suboptimal}. On the other hand, the  social planner can significantly lower the total cost when the budget increases, as she puts all the budget only on the min-cut edges, as suggested by Proposition~\ref{prop: non-behavioral min cut}; this reduces the total cost faster towards zero as the budget increases.

Other practical scenarios (such as deploying moving-target defense) where our results are applicable can be found in the book chapter \cite{hota2018game}. While our results show that the inefficiency becomes exacerbated as the total budget increases, this property does not hold for all networks. We omit further discussions about these aspects due to space constraints.


\begin{figure}
\begin{center}
  \includegraphics[width=0.6\columnwidth]{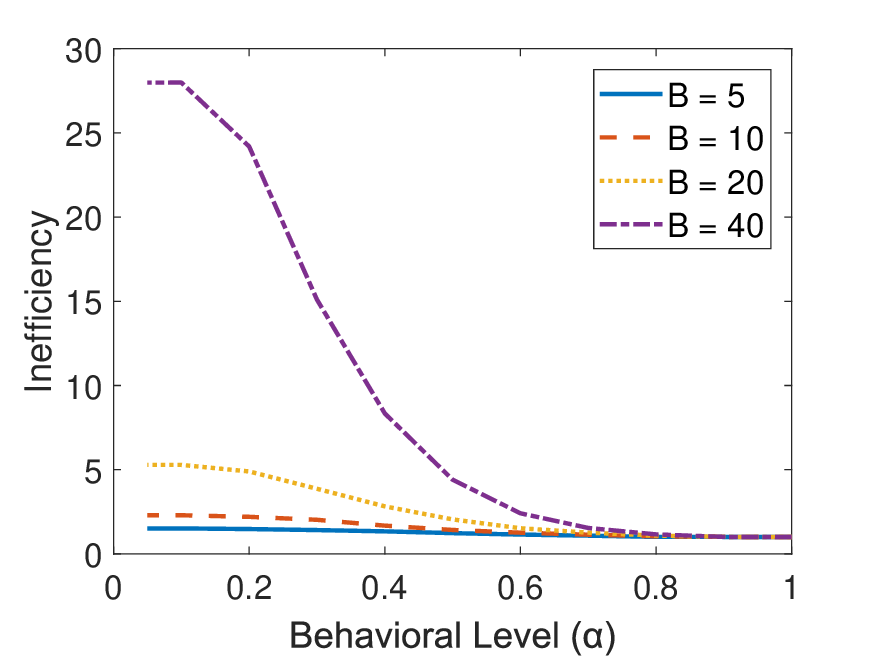}
  \caption{The inefficiency for different behavioral levels of the defenders. We observe that the inefficiency increases as the security budget increases, and as the defenders become more behavioral.\protect\footnotemark}
  \label{fig:PoBA}
\end{center}
\end{figure}

\footnotetext{Recall that the inefficiency is the ratio of the total system true expected cost at a PNE to the total system true expected cost at the (non-behavioral) social optimum.}

%% file: Conclusion.tex
\section{Summary of Findings}\label{sec:conclusion}
In this paper, we presented an analysis of the impacts of behavioral decision-making on the security of interdependent systems.  First, we showed that the optimal investments by a behavioral decision-maker will be unique, whereas non-behavioral decision-makers may have multiple optimal solutions. Second, non-behavioral decision-makers find it optimal to concentrate their security investments on minimum edge-cuts in the network in order to protect their assets, whereas behavioral decision-makers will choose to spread their investments over other edges in the network, potentially making their assets more vulnerable.  Third, we showed that multi-defender games possess a PNE (under appropriate conditions on the game), and introduced a metric that we termed the ``Price of Behavioral Anarchy'' to quantify the inefficiency of the PNE as compared to the security outcomes under socially optimal investments.  We provided a tight bound on PoBA, which depended only on the total budget across all defenders.  However, we also showed that the tendency of behavioral defenders to spread their investments over the edges of the network can potentially benefit the other defenders in the network.  Finally, we presented a case study where the inefficiency of the equilibrium increased as the defenders became more behavioral.

In total, our analysis shows that human decision-making (as captured by behavioral probability weighting) can have substantial impacts on the security of interdependent systems, and must be accounted for when designing and operating distributed, interdependent systems. In other words, the insights that are provided by our work (e.g., that behavioral decision-makers may move some of their security investments away from critical portions of the network) can be used by system planners to identify portions of their network that may be left vulnerable by the human security personnel who are responsible for managing those parts of the network. A future avenue for research is to perform human experiments to test our predictions. Moreover, studying the properties of security investments when different edges have different degrees of misperception of attack probabilities is another avenue for future research.

